\documentclass[11pt,reqno]{amsart}
\allowdisplaybreaks[4]
\usepackage{graphicx} 
\usepackage{epsfig}
\usepackage{amssymb}
\usepackage{amsmath}
\usepackage{cite}
\newtheorem{theorem}{Theorem}

\newtheorem{proposition}[theorem]{Proposition}

\newtheorem{lemma}[theorem]{Lemma}
\newcommand{\be}{\begin{equation}}
\newcommand{\ee}{\end{equation}}
\newcommand{\bea}{\begin{eqnarray}}
\newcommand{\eea}{\end{eqnarray}}
\newcommand{\ba}{\begin{array}}
\newcommand{\ea}{\end{array}}
\newcommand{\bean}{\begin{eqnarray*}}
\newcommand{\eean}{\end{eqnarray*}}

\newcommand{\pa}{\partial}

\begin{document}

\title{The applications of the gauge transformation\\
 for the BKP hierarchy}
\author{Jipeng Cheng\dag, Jingsong He$^*$ \ddag }
\dedicatory {  \dag \ Department of Mathematics, China University of
Mining and Technology, Xuzhou, Jiangsu 221116,
P.\ R.\ China\\
\ddag \ Department of Mathematics, Ningbo University, Ningbo,
Zhejiang 315211, P.\ R.\ China }

\thanks{$^*$Corresponding author. Email: hejingsong@nbu.edu.cn.}
\begin{abstract}
In this paper, we investigated four applications of the gauge
transformation for the BKP hierarchy. Firstly, it is found that the
orbit of the gauge transformation for the constrained BKP hierarchy
defines a special $(2 +1)$-dimensional Toda lattice equation
structure. Then the tau function of the BKP hierarchy generated by
the gauge transformation is showed to be the Pfaffian. And the
higher Fay-like identities for the BKP hierarchy is also obtained
through the gauge transformation. At last, the compatibility between
the additional symmetry and the gauge transformation of the BKP
hierarchy is proved.
\end{abstract}
\maketitle
\section{Introduction}
The gauge transformation \cite{chau1992,oevel1993} provides a simple way to construct solutions for integrable hierarchies.
By now, the gauge transformations of many integrable hierarchies have been constructed, for example, the KP hierarchy\cite{chau1992,oevel1993},
the constrained KP hierarchy\cite{oevel1993,aratyn1995,chau1997,willox,he2003}, the constrained BKP and CKP hierarchy\cite{nimmo,he2007,hejhep,hejnmp}
 (cBKP and cCKP), the discrete KP hierarchy\cite{oevel1996,liu2010}, the $q$-KP hierarchy\cite{tu1998,he2006} and so on.
 In this paper, we will mainly study the gauge transformation for the BKP hierarchy.

 The BKP hierarchy \cite{djkm1983} is one of the most important sub-hierarchies of the
  KP hierarchy defined by restricting the Lax operator
\footnote{Here $\pa=\pa_x$ and the asterisk stands for the conjugation
operation: $(AB)^*=B^*A^*$, $\pa^*=-\pa$, $f^*=f$ with $f$ be a
function.}$L^*=-\pa L \pa^{-1}$. In order to keep the restrictions
on the Lax operator of the BKP hierarchy, one cannot do gauge
transformation for the BKP hierarchy by only using one kind of the
basic gauge transformation operators: the differential type $T_D$
and the integral type $T_I$ developed by Chau {\it et al} in
\cite{chau1992}, and instead has to use the combination of $T_D$ and
$T_I$\cite{nimmo,he2007,hejhep,hejnmp}. In particular, the generating functions
of the combination of $T_D$ and $T_I$ are dependent
on each other, which is also requested by the restriction of the Lax operator.
 As for the cBKP hierarchy, besides the usual BKP constraint on the Lax operator,
another constraint defined with the eigenfunctions (see (\ref{cbkplax})) is also needed.
So to ensure this additional  constraint for the cBKP hierarchy, the corresponding
gauge transformation operator can be the one of the usual BKP hierarchy just by
letting the generating function be one of the original eigenfunctions in the definition.

Furthermore, it is an interesting problem to explore the non-trivial
relevance of the gauge transformation to the other integrable
properties of the BKP besides the explicit solutions, which shall be
illustrated from the following four concerns. We firstly show that the
orbit of the gauge transformation for the constrained BKP hierarchy
defines a special $(2 +1)$-dimensional Toda lattice equation
structure (see (\ref{2todaequation})), proposed by Cao {\it et al}
in \cite{cao}. This equation has some importance in mathematics and
physics, whose continuous analogue is equivalent to the Ito equation
\cite{qianjpa}. And some interesting integrable properties of this
special $(2 +1)$-dimensional Toda lattice equation can be seen in
\cite{cao,qianjpa, hujmp,qiaozj}.

Then starting from the Grammian determinant solutions of the BKP hierarchy, we derived
two types of the Pfaffian structures \cite{hirotajpsj,lorisjmp1999} for the BKP hierarchy (see Proposition \ref{bkpgaugen}).
Since the BKP hierarchy is the subhierarchy of the KP hierarchy, there are two types of the tau functions for the BKP hierarchy\cite{djkm1983}:
one is inherited from the KP hierarchy (denoted by $\tau_{KP}$) by cancelling the even flow variables, the other is of its own (denoted by $\tau$) defined directly by the
odd variables according to the flow equations. In \cite{he2007},  the transformed tau function under the gauge transformation of the BKP hierarchy is only provided for $\tau_{KP}$
in the form of the Grammian determinant\cite{lorisjmp1997}, while the transformed tau function of its own is not considered. The crucial point to derive the
transformed tau function of BKP's own is to root the Grammian determinant according to the relation of these two types
of the tau functions\cite{djkm1983}. In this paper, we derived the two types of the Pfaffian structures for the BKP hierarchy, according to the even or odd
steps of the gauge transformation. In particular, for the case of the even times, it can be obtained directly by using the definition of the Pfaffian,
while the case of the odd times is some complicated.

Further through the gauge transformation of the BKP hierarchy, the higher order Fay-like identities\cite{vM94,vandeluer,Tu07} are derived. The transformed tau function
under the $k+1$-step gauge transformation can be expressed in two different ways. One is starting from the spectral representation for the eigenfunction
of the BKP hierarchy \cite{cheng2010}, then after the application of the $k$-step gauge transformation, the recurrence relation for the transformed tau function
can be obtained. The first expression is derived by solving this recurrence relation. The other is got by expressing the eigenfunctions in the Pfaffian structure for the transformed tau function through the spectral representation\cite{cheng2010}. The comparison of these two different expressions gives rise to the higher Fay-like identities.

At last, the compatibility between the additional symmetry and the gauge transformation of the BKP
hierarchy is studied. The additional symmetry\cite{OS86,ASM95,D95,takasaki1993,Tu07,he2007b,cheng2011,1li2012,tian2011}
is a kind of symmetry depending explicitly on the space and time variables, involved in so-called string equation and the generalized
Virasoro constraints in matrix models of the 2d quantum gravity (see \cite{vM94} and references therein).  In order to show a new inner
consistency of the BKP integrable hierarchy, it is an interesting problem to investigate the
compatibility between the gauge transformation and the additional symmetry.

This paper is organized in the following way. In section 2, some backgrounds about the BKP hierarchy are presented.
Then, the gauge transformations of the BKP and cBKP hierarchies are reviewed
and the orbits of the gauge transformation for BKP hierarchy are studied in Section 3. Then, the Pfaffian structures
of the BKP hierarchy are investigated in Section 4. Further, the higher order Fay-like identities through the gauge transformation
is derived in Section 5. And in Section 6, the compatibility of the additional symmetry and the gauge transformation for the BKP hierarchy is checked.
At last, we devote section 7 to some conclusions and discussions.
\section{Backgrounds on the BKP Hierarchy}
In this section, we shall review some backgrounds of the BKP
hierarchy \cite{djkm1983}. The BKP hierarchy is defined in Lax form
as follows
\begin{equation}\label{laxeq}
\pa_{2n+1}L=[B_{2n+1},L],\quad B_{2n+1}=(L^{2n+1})_+,\quad
n=0,1,2,\cdots,
\end{equation}
 where the Lax operator is given by
 \begin{equation}\label{laxop}
 L=\pa+u_2\pa^{-1}+u_3\pa^{-2}+\cdots,
\end{equation}
with the coefficient functions $u_i$ depending on the time variables
$t=(t_1=x,t_3,t_5,\cdots)$, and satisfies the
 BKP constraint
\begin{equation} \label{constr}
L^*=-\pa L\pa^{-1},
\end{equation}
which is equivalent to the condition
\begin{equation} \label{constrequivalent}
B_{2n+1}(1)=0.
\end{equation}
 Here
$\pa_{2n+1}=\pa_{t_{2n+1}}$. And for any (pseudo-) differential
operator $A$, $B$, and a function $f$, $(A)_{\pm}$ denote the
differential part and the integral part of the pseudo-differential
operator $A$ respectively. The symbol $A(f)$ will indicate the
action of $A$ on $f$, whereas the symbol $Af$ will denote just
operator product of $A$ and $f$. The Lax equation (\ref{laxeq}) is
equivalent to the compatibility condition of the linear system
\begin{equation}\label{lineq}
L(\psi_{BA}(t,\lambda))=\lambda \psi_{BA}(t,\lambda),\qquad
\pa_{2n+1}\psi_{BA}(t,\lambda)=B_{2n+1}(\psi_{BA}(t,\lambda)),
\end{equation}
where $\psi_{BA}(t,\lambda)$ is called the Baker-Akhiezer (BA) wave
function.

The whole hierarchy can be expressed in terms of a dressing operator
$W$, so that
\begin{equation}\label{laxdressing}
L=W\pa W^{-1},\qquad W=1+\sum_{j=1}^\infty w_j\pa^{-j},
\end{equation}
and the Lax equation is equivalent to the Sato's equation
\begin{equation}\label{satoeq}
\pa_{2n+1}W=-(L^{2n+1})_-W,
\end{equation}
with constraint
 \begin{equation}\label{constrsato}
  W^*\pa W=\pa.
 \end{equation}

Let the solutions of the linear system (\ref{lineq}) be the form
\begin{equation}\label{wavefun}
\psi_{BA}(t,\lambda)=W(e^{\xi(t,\lambda)})=w(t,\lambda)e^{\xi(t,\lambda)}
\end{equation}
where $\xi(t,\lambda)=\sum_{i=0}^\infty t_{2i+1}\lambda^{2i+1}$ and
$w(t,\lambda)=1+w_1/\lambda+w_2/\lambda^2+\cdots$. Then
$\psi_{BA}(t,z)$ is a wave function of the BKP hierarchy if and only
if it satisfies the bilinear identity\cite{djkm1983} \be \int
d\lambda\lambda^{-1}\psi_{BA}(t,\lambda)\psi_{BA}(t',-\lambda)=1,\quad
\forall t, t', \label{bileq} \ee where $\int
d\lambda\equiv\oint_\infty\frac{d\lambda}{2\pi i}=
Res_{\lambda=\infty}$ and $t=(t_1=x,t_3,t_5, \cdots)$.

From the bilinear identity (\ref{bileq}), solutions of the BKP
hierarchy can be characterized by a single function $\tau(t)$ called
$\tau$-function such that\cite{djkm1983}
\begin{equation}\label{tau}
\psi_{BA}(t,\lambda)=\frac{\tau
(t-2[\lambda^{-1}])}{\tau(t)}e^{\xi(t,\lambda)},
\end{equation}
where
$[\lambda^{-1}]=(\lambda^{-1},\frac{1}{3}\lambda^{-3},\cdots)$. This
implies that all dynamical variables $\{u_i \}$ in the Lax  operator
$L$ can be expressed by $\tau$-function. Moreover, another important
property of $\tau$ function of the BKP is the following Fay like
identity\cite{Tu07}
 \begin{eqnarray}
&&\sum_{(s_1,s_2,s_3)}\frac{(s_1-s_0)(s_1+s_2)(s_1+s_3)}{(s_1+s_0)(s_1-s_2)(s_1-s_3)}
\tau(t+2[s_2]+2[s_3])\tau(t+2[s_0]+2[s_1])\nonumber\\
&&+\frac{(s_0-s_1)(s_0-s_2)(s_0-s_3)}{(s_0+s_1)(s_0+s_2)(s_0+s_3)}
\tau(t+2[s_0]+2[s_1]+2[s_2]+2[s_3])\tau(t)=0, \label{Fayeq}
\end{eqnarray}
where $(s_1,s_2,s_3)$ stands for cyclic permutations of $s_1$, $s_2$
and $s_3$, and the differential Fay identity \cite{Tu07}
 \begin{eqnarray}
&&\left(\frac{1}{s_2^2}-\frac{1}{s_1^2}\right)
\{\tau(t+2[s_1])\tau(t+2[s_2])-\tau(t+2[s_1]+2[s_2])\tau(t)\}\nonumber\\
&&=\left(\frac{1}{s_2}+\frac{1}{s_1}\right)
\{\pa\tau(t+2[s_2])\tau(t+2[s_1])-\pa\tau(t+2[s_1])\tau(t+2[s_2])\}\nonumber\\
&&\quad+\left(\frac{1}{s_2}-\frac{1}{s_1}\right)\{\tau(t+2[s_1]+2[s_2])\pa\tau(t)-
\pa\tau(t+2[s_1]+2[s_2])\tau(t)\}. \label{dFayeq}
\end{eqnarray}

In the BKP hierarchy, if $\Phi$ (or $\Psi$) satisfies
\begin{equation}\label{eigen}
    \pa_{2n+1}\Phi=B_{2n+1}(\Phi)
    \ \left({\rm or}\ \pa_{2n+1}\Psi=-B^*_{2n+1}(\Psi)\right),\ n=0,1,2,\cdots,
\end{equation}
we shall call $\Phi$ (or $\Psi$) eigenfunction (or adjoint
eigenfunction) of the BKP hierarchy. Obviously, by
(\ref{constrequivalent}) and (\ref{lineq}), $1$ and
$\psi_{BA}(t,\lambda)$ are the eigenfunctions.  From the fact
$B^*_{2n+1}\pa=-\pa B_{2n+1}$, any adjoint eigenfunction $\Psi$ can
be given in the form of $\Psi=\Phi_x$ with $\Phi$ be an
eigenfunction. In particular, the adjoint BA function
$\psi^*_{BA}(t,\lambda)= -\lambda^{-1}\psi_{BA}(t,-\lambda)_x$,
where
\begin{equation}\label{adjointbafunction}
\psi^*_{BA}(t,\lambda)\equiv W^{*-1}(e^{-\xi(t,\lambda)}).
\end{equation}
Thus in the BKP case, it is enough to only consider the
eigenfunctions. The relation between the eigenfunction $\Phi$ and
the BA wave function $\psi_{BA}(t,\lambda)$ is showed in the
following spectral representation \cite{cheng2010},
\begin{equation}\label{spectralrepresentation}
    \Phi(t)=\int d\lambda \psi_{BA}(t,\lambda)\varphi(\lambda),
\end{equation}
with
$\varphi(\lambda)=\lambda^{-1}S\left(\Phi(t'),\psi_{BA}(t',-\lambda)_{x'}\right)$,
where for any pair of (adjoint) eigenfunctions $\Phi(t),\Psi(t)$,
$S(\Phi(t),\Psi(t))$ is determined by the following equations,
\begin{equation}\label{sep1}
    \frac{\pa}{\pa t_{2n+1}}S\big(\Phi(t),\Psi(t)\big)=Res\left(\pa^{-1}\Psi B_{2n+1}\Phi\pa^{-1}\right),
n=0,1,2,3,\cdots.
\end{equation}
In particular for $n=0$,
\begin{equation}\label{sep2}
    \pa_xS(\Phi(t),\Psi(t))=\Phi(t)\Psi(t).
\end{equation}

 The constrained BKP hierarchy \cite{cheng1992,lorisjmp1999} can be defined by restricting the
Lax operator of the BKP hierarchy (\ref{laxeq}) in the following
form:
\begin{equation}\label{cbkplax}
    L^{2k+1}=\pa^{2k+1}+\sum_{i=1}^{2k-1}v_i\pa^{i}+\sum_{j=1}^m\left(\Phi_{1j}\pa^{-1}\Phi_{2j,x}-\Phi_{2j}\pa^{-1}\Phi_{1j,x}\right),
\end{equation}
for some fixed $k$, where $\Phi_{1j}$ and $\Phi_{2j}$ are the
eigenfunctions satisfying (\ref{eigen}).

\section{the Gauge Transformation for the BKP Hierarchy}
Now let's review some results about the gauge transformation for the
BKP hierarchy\cite{nimmo,he2007,hejhep,hejnmp}. Assume $T$ be a pseduo-differential operator and
$L^{(1)}=TL^{(0)}T^{-1}$ with $L^{(0)}$ be the Lax operator of the
BKP hierarchy. $T$ is called the gauge transformation for the BKP
hierarchy, if it satisfies:
\begin{itemize}
  \item preserving the Lax equation: $\pa_{2n+1} L^{(1)}=[B_{2n+1}^{(1)},L^{(1)}],\quad
  B_{2n+1}^{(1)}=(L^{(1)})^{2n+1}_+,$
  \item preserving the BKP constraint: $L^{(1)*}=-\pa L^{(1)}\pa^{-1}$.
\end{itemize}

In \cite{nimmo,he2007,hejhep}, it is showed that the operator
$T(\Phi)=T_I(\Phi/2)T_D(\Phi)$ satisfies the two conditions above,
and thus is the gauge transformation operator of the BKP hierarchy,
where $\Phi$ is the eigenfunctions of the BKP hierarchy and
\begin{equation}\label{tid}
    T_D(\Phi)=\Phi\pa \Phi^{-1},\quad T_I(\Phi/2)=(\Phi/2)^{-1}\pa^{-1}
    (\Phi/2).
\end{equation}
There is an important property satisfied by
$T(\Phi)=T_I(\Phi/2)T_D(\Phi)$ listed in the lemma below \cite{hejhep}.
\begin{lemma} The operator $T(\Phi)=T_I(\Phi/2)T_D(\Phi)$ obeys
   \begin{equation}\label{tproperty}
    T(\Phi)^*\pa T(\Phi)=\pa.
   \end{equation}
\end{lemma}

Under the transformation $T(\Phi)=T_I(\Phi/2)T_D(\Phi)$, we will
have
\begin{eqnarray}
L^{(0)}\rightarrow L^{(1)}&=&T(\Phi)L^{(0)}T(\Phi)^{-1},\label{bkplaxgauge1}\\
\chi^{(0)}\rightarrow\chi^{(1)}&=&T(\Phi)(\chi^{(0)}),\label{bkpeigengauge1}\\
\tau_{KP}^{(0)}\rightarrow\tau_{KP}^{(1)}&=&\frac{\Phi^2}{2}
\tau_{KP}^{(0)},\label{bkptaukpgauge1}
\end{eqnarray}
where $L^{(0)}$ is the initial Lax operator of the BKP hierarchy,
$\chi^{(0)}$ is the initial eigenfunction, and $\tau_{KP}^{(0)}$ is
the initial BKP tau function inherited from the KP hierarchy, which
satisfies \cite{djkm1983}
\begin{equation}\label{taurelation}
    \tau_{KP}=\tau^2.
\end{equation}
From (\ref{bkptaukpgauge1}) and (\ref{taurelation}),
\begin{eqnarray}
\tau^{(0)}\rightarrow\tau^{(1)}&=&\frac{\Phi}{\sqrt{2}}
\tau^{(0)}.\label{bkptaubkpgauge1}
\end{eqnarray}
The actions of the gauge transformation
$T(\Phi)=T_I(\Phi/2)T_D(\Phi)$ on the dressing operator and the BA
function are showed as follows.
\begin{eqnarray}
W^{(0)}\rightarrow W^{(1)}&=&T(\Phi)W^{(0)},\label{bkpdressinggauge1}\\
\psi_{BA}^{(0)}\rightarrow\psi_{BA}^{(1)}&=&T(\Phi)(\psi_{BA}^{(0)}),\label{bkpwavegauge1}
\end{eqnarray}

Therefore under the successive gauge transformation
\begin{equation}\label{tngaugebkp}
T_n=T(\Phi_n^{(n-1)})T(\Phi_{n-1}^{(n-2)})...T(\Phi_2^{(1)})T(\Phi_1^{(0)}),
\end{equation}
where \{$\Phi_1,\Phi_2,...,\Phi_n$\} are a group of different
eigenfunctions of the BKP hierarchy, the Lax operator $L^{(0)}$ and
the BKP tau function $\tau^{(0)}$ will transform as
\begin{eqnarray}
L^{(n)}&=&T_nL^{(0)}T_n^{-1},\label{bkplaxgaugen}\\
\tau_{KP}^{(n)}&=&G_n(\Phi_1,\Phi_2,...,\Phi_n)
\tau_{KP}^{(0)},\label{bkptaukpgaugen}\\
\tau^{(n)}&=&\left(\frac{1}{2}\right)^{\frac{n}{2}}\Phi_n^{(n-1)}\Phi_{n-1}^{(n-2)}...\Phi_2^{(1)}\Phi_1^{(0)}
\tau^{(0)},\label{bkptaubkpgaugen}
\end{eqnarray}
where
\begin{equation}\label{gnexpression}
    G_n(\Phi_1,\Phi_2,...,\Phi_n)={\rm det}\left(S(\Phi_i,\Phi_{jx})\right)_{1\leq i,j\leq n}.
\end{equation}

The gauge transformation for the cBKP hierarchy is constructed in
\cite{he2007} from the corresponding results about the cKP hierarchy by
considering the BKP constraint. Now let's review the results about
the gauge transformation for the cBKP hierarchy.

Assume
\begin{equation}\label{cbkplax0}
    \left(L^{(0)}\right)^{2k+1}=\left(L^{(0)}\right)^{2k+1}_++\sum_{i=1}^m\left(\Phi_{1i}^{(0)}\pa^{-1}\left(\Phi_{2i}^{0)}\right)_x
  -\Phi_{2i}^{(0)}\pa^{-1}\left(\Phi_{1i}^{(0)}\right)_x\right)
\end{equation}
for the initial Lax operator of the cBKP hierarchy. Under the BKP
gauge transformation operator $T(\chi)=T_I(\chi)T_D(\chi)$,
(\ref{cbkplax0}) will become into:
\begin{eqnarray}
  \left(L^{(1)}\right)^{2k+1}&=&T(\chi)\left(L^{(0)}\right)^{2k+1}T_I(\chi)^{-1}=\left(L^{(1)}\right)^{2k+1}_++\left(L^{(1)}\right)^{2k+1}_-,\label{cbkplax1}\\
  \left(L^{(1)}\right)^{2k+1}_-&=&\Phi_{1,0}^{(1)}\pa^{-1}\left(\Phi_{2,0}^{1)}\right)_x
  -\Phi_{2,0}^{(1)}\pa^{-1}\left(\Phi_{1,0}^{(1)}\right)_x\nonumber\\
  &&+\sum_{i=1}^m\left(\Phi_{1i}^{(1)}\pa^{-1}\left(\Phi_{2i}^{1)}\right)_x
  -\Phi_{2i}^{(1)}\pa^{-1}\left(\Phi_{1i}^{(1)}\right)_x\right),\label{cbkplax1minus}\\
  \Phi_{1,0}^{(1)}&=&\left(T(\chi)\left(L^{(0)}\right)^{2k+1}\right)\left(\Phi_{1,0}^{(0)}\right),\quad \Phi_{2,0}^{(1)}=-\frac{2}{\Phi_{1,0}^{(0)}},\label{phi01}\\
  \Phi_{1i}^{(1)}&=&T(\chi)\left(\Phi_{1i}^{(0)}\right),\quad
  \Phi_{2i}^{(1)}=T(\chi)^{*-1}\left(\Phi_{2i}^{(0)}\right),\label{phii1}
\end{eqnarray}
In order to preserve the form of (\ref{cbkplax0}), $\chi$ is
required to coincide with one of the original eigenfunctions in
$\left(L^{(0)}\right)^{2k+1}$, e.g. $\chi=\Phi_{1,1}^{(0)}$, since
$\Phi_{1,1}^{(1)}=0$ in this case. Applying successive the BKP gauge
transformations
\begin{equation}\label{cbkplaxn}
    L^{(n+1)}=T^{(n)}L^{(n)}\left(T^{(n)}\right)^{-1},\ \
    T^{(n)}\equiv T(\Phi_{1,1}^{(n)})
\end{equation}
yields:
\begin{eqnarray}
  \left(L^{(n+1)}\right)^{2k+1}&=&\left(L^{(n+1)}\right)^{2k+1}_++\left(L^{(n+1)}\right)^{2k+1}_-,\label{cbkplaxnaddmin}\\
  \left(L^{(n+1)}\right)^{2k+1}_-&=&\sum_{i=1}^m\left(\Phi_{1i}^{(n+1)}\pa^{-1}\left(\Phi_{2i}^{n+1)}\right)_x
  -\Phi_{2i}^{(n+1)}\pa^{-1}\left(\Phi_{1i}^{(n+1)}\right)_x\right),\label{cbkplaxnmin}\\
  \Phi_{1,1}^{(n+1)}&=&\left(T^{(n)}\left(L^{(n)}\right)^{2k+1}\right)\left(\Phi_{1,1}^{(n)}\right),\quad \Phi_{2,1}^{(n+1)}=-\frac{2}{\Phi_{1,1}^{(n)}},\label{phi1n}\\
  \Phi_{1i}^{(n+1)}&=&T^{(n)}\left(\Phi_{1i}^{(n)}\right),\quad
  \Phi_{2i}^{(n+1)}=\left(T^{(n)}\right)^{*-1}\left(\Phi_{2i}^{(n)}\right),i=2,3,4,...\label{phiin}\\
  \tau^{(n+1)}&=&\frac{\Phi_{1,1}^{(n)}}{\sqrt{2}}\tau^{(n)}.\label{btaun}
\end{eqnarray}

We next restrict the cBKP hierarchy (\ref{cbkplax}) to $m=1$ case
and then will find the following proposition:
\begin{proposition}\label{orbitgauge}
The orbit of the gauge transformation for the cBKP hierarchy
(\ref{cbkplax}) (for $k\geq 1$ and $m=1$) defines a special $(2 +
1)$-dimensional Toda lattice equation structure (here we set
$\Phi=\Phi_{1,1}$).
\begin{eqnarray}
\pa_x\pa_{2k+1}{\rm
ln}\Phi^{(n)}=\frac{\left(\Phi^{(n)}\Phi^{(n+1)}\right)_x}{\left(\Phi^{(n)}\right)^2}
+\left(\left(\Phi^{(n-1)}\right)^{-1}\right)_x\Phi^{(n)}-\left(\Phi^{(n)}\right)_x\left(\Phi^{(n-1)}\right)^{-1},n=0,1,2,...\label{2toda}
\end{eqnarray}

\end{proposition}
\begin{proof}
Firstly from (\ref{cbkplaxnmin}) and (\ref{phi1n}),
\begin{eqnarray*}
\Phi_{1,1}^{(n+1)} &=&
T^{(n)}\left(\left(L^{(n)}\right)_+^{2k+1}\left(\Phi_{1,1}^{(n)}\right)\right)+T^{(n)}\left(\left(L^{(n)}\right)_-^{2k+1}\left(\Phi_{1,1}^{(n)}\right)\right)\nonumber\\
&=&
T^{(n)}\left(\pa_{2k+1}\left(\Phi_{1,1}^{(n)}\right)\right)+T^{(n)}\left(\Phi_{1,1}^{(n)}\int^x\left(\Phi_{2,1}^{(n)}\right)_x\Phi_{1,1}^{(n)}
-\Phi_{2,1}^{(n)}\int^x\left(\Phi_{1,1}^{(n)}\right)_x\Phi_{1,1}^{(n)}\right)
\end{eqnarray*}
Then multiplying $\Phi_{1,1}^{(n)}$ to the both sides of the above identity
and differentiating it with respect to $x$, by the definition of $T^{(n)}$, we have
\begin{eqnarray*}
\frac{\left(\Phi_{1,1}^{(n)}\Phi_{1,1}^{(n+1)}\right)_x}{\left(\Phi_{1,1}^{(n)}\right)^2}
&=&
\pa_x\pa_{2k+1}{\rm ln}\left(\Phi_{1,1}^{(n)}\right)+\left(\Phi_{2,1}^{(n)}\right)_x\Phi_{1,1}^{(n)}-\frac{1}{2}\pa_x\left(\Phi_{2,1}^{(n)}\Phi_{1,1}^{(n)}\right)\nonumber\\
&=&
\pa_x\pa_{2k+1}{\rm ln}\left(\Phi_{1,1}^{(n)}\right)+\frac{1}{2}\left(\Phi_{2,1}^{(n)}\right)_x\Phi_{1,1}^{(n)}-\frac{1}{2}\left(\Phi_{1,1}^{(n)}\right)_x\Phi_{2,1}^{(n)}\nonumber\\
&=& \pa_x\pa_{2k+1}{\rm
ln}\left(\Phi_{1,1}^{(n)}\right)-\left(\left(\Phi_{1,1}^{(n-1)}\right)^{-1}\right)_x\Phi_{1,1}^{(n)}+\left(\Phi_{1,1}^{(n)}\right)_x\left(\Phi_{1,1}^{(n-1)}\right)^{-1}
\end{eqnarray*}
At last by considering $\Phi=\Phi_{1,1}$, (\ref{2toda}) will be
derived.
\end{proof}
Note that (\ref{2toda}) if set $\Phi^{(n)}=e^{u(n)}$ and
$t_{2k+1}=y$, then we will get the following special $(2 +
1)$-dimensional Toda lattice equation \cite{cao,qianjpa, hujmp,qiaozj}
\begin{equation}\label{2todaequation}
\pa_x\pa_y
u(n)=e^{u(n+1)-u(n)}\pa_x(u(n+1)+u(n))-e^{u(n)-u(n-1)}\pa_x(u(n)+u(n-1)).
\end{equation}

\section{the Pfaffian Structure of the BKP Hierarchy}
In this section, we will investigate the Pfaffian structures of the BKP
hierarchy generated by the gauge transformation.

Firstly, let's denote \cite{lorisjmp1999,cheng2010}
\begin{equation}\label{bsep}
    \Omega(\Phi_1,\Phi_2)=S(\Phi_2,\Phi_{1x})-S(\Phi_1,\Phi_{2x})
\end{equation}
for the eigenfunctions $\Phi_1$ and $\Phi_2$ of the BKP hierarchy,
which satisfies
\begin{equation}\label{bsepproperty}
    \Omega(\Phi_1,\Phi_2)=-\Omega(\Phi_2,\Phi_1).
\end{equation}
In particular, $\Omega(\Phi_1,1)=\Phi_1$, since $1$ is also an
eigenfunction of the BKP hierarchy. With the help of the spectral
representation (\ref{spectralrepresentation}) and (\ref{tau}), we
have \cite{cheng2010}
\begin{equation}\label{bsepgeneralformula}
    \Omega(\Phi_1,\Phi_2)=\int\int d\lambda d\mu
\varphi_1(\lambda)\varphi_2(\mu)e^{\xi(t,\lambda)+\xi(t,\mu)}\frac{\lambda-\mu}{\lambda+\mu}\frac{\tau(t-2[\lambda^{-1}]-2[\mu^{-1}])}{\tau(t)}.
\end{equation}

Thus from (\ref{sep2}) and (\ref{bsep}),
\begin{equation}\label{bseppropertysep}
    S(\Phi_1,\Phi_{2x})=\frac{1}{2}\left(\Phi_1\Phi_2-\Omega(\Phi_1,\Phi_2)\right).
\end{equation}

Then we start from the $\tau_{KP}^{(n)}$ generated by the gauge
transformation $T_n$ (see (\ref{tngaugebkp})) for the BKP hierarchy.
According to (\ref{bsepproperty}) and (\ref{bseppropertysep}), we
find
\begin{equation}\label{bkptaubkpgaugen2}
    \frac{\tau_{KP}^{(n)}}{\tau_{KP}^{(0)}}=\frac{1}{2^n}{\rm det}\left(\Phi_i\Phi_j-\Omega(\Phi_i,\Phi_j)\right)_{1\leq i,j\leq n}
    =\frac{1}{2^n}{\rm det}\left(\Phi_i\Phi_j+\Omega(\Phi_i,\Phi_j)\right)_{1\leq i,j\leq
    n},
\end{equation}
where ${\rm det}(A)={\rm det}(A^T)$ ($T$ denote the transpose) for
some matrix $A$ has been used. Before the further discussion, the
lemma \cite{wangqingwen} below is needed.
\begin{lemma}
If set $A=(a_{i,j})_{1\leq i,j n}$, $A^*$ be the adjoint matrix of
$A$, and
$$X=(x_1,x_2,...,x_n)^T,\quad Y=(y_1,y_2,...,y_n)^T,$$
then
\begin{equation}\label{matrixident}
{\rm det}(A+XY^T)={\rm det}(A)+Y^TA^*X.
\end{equation}
\end{lemma}
Thus if denote $X=(\Phi_1,\Phi_2,..., \Phi_n)^T$ and
$\Omega=(\Omega(\Phi_i,\Phi_j))_{1\leq i,j\leq
n}\triangleq(\Omega_{ij})_{1\leq i,j\leq n}$, then from
(\ref{bkptaubkpgaugen2}) and (\ref{matrixident}), we have
\begin{equation}\label{bkptaubkpgaugen3}
\frac{2^n\tau_{KP}^{(n)}}{\tau_{KP}^{(0)}}= {\rm
det}(\Omega)+X^T\Omega^*X=\left\{
                            \begin{array}{ll}
                              {\rm det}(\Omega), & \hbox{\text{$n$ is even};} \\
                              X^T\Omega^*X, & \hbox{\text{$n$ is odd};}
                            \end{array}
                          \right.
\end{equation}
since $\Omega$ is antisymmetric matrix, and when $n$ is odd,
$\Omega^*$ is symmetric matrix, while $n$ is even, $\Omega^*$ is
antisymmetric. Therefore according to (\ref{taurelation}), the transformed form of the BKP's own tau
function $\tau$ can be derived just by taking the square root of (\ref{bkptaubkpgaugen3}).

When $n$ is even, the transformed BKP's own tau
function $\tau$ can be expressed as the Pfaffian \cite{hirotajpsj,lorisjmp1999},
 which is defined as the square root of the determinat
of an antisymmetric matrix of even order (see \cite{hirotabook} for more details). For example, the Pfaffians
of $2\times2$ and $4\times 4$ antisymmetric matrix $A=(a_{ij})$ are
${\rm Pf}(A)=a_{12}$ and ${\rm
Pf}(A)=a_{12}a_{34}-a_{13}a_{24}+a_{14}a_{23}$. Thus according to
(\ref{taurelation}) and (\ref{bkptaubkpgaugen3}),
\begin{equation}\label{bkptaubkpgaugeneven}
\frac{\tau^{(n)}}{\tau^{(0)}}= \frac{(-1)^{a(n)}}{2^{n/2}}{\rm
Pf}(\Phi_1,\Phi_2,...,\Phi_n),
\end{equation}
where ${\rm Pf}(\Phi_1,\Phi_2,...,\Phi_n)={\rm Pf}(\Omega)$. The
choice of the sign of (\ref{bkptaubkpgaugeneven}), when taking
square root, can be seen as follows. According to
(\ref{bkptaubkpgaugen}), if letting $\Phi_{n+1}=1$, since $1$ is
also the eigenfunction of the BKP hierarchy, and noting that
$1^{(n)}=T_n(1)=1^n=1$, then
\begin{equation}\label{evensign1}
\frac{\tau^{(n)}}{\tau^{(0)}}=\sqrt{2}\frac{\tau^{(n+1)}}{\tau^{(0)}},
\end{equation}
thus $\frac{\tau^{(n)}}{\tau^{(0)}}$ and
$\frac{\tau^{(n+1)}}{\tau^{(0)}}$ have the same sign. Similarly, if
setting $\Phi_{n+2}=1$, by noting that $1^{(n+1)}=-1$, then
\begin{equation}\label{evensign2}
\frac{\tau^{(n+1)}}{\tau^{(0)}}=-\sqrt{2}\frac{\tau^{(n+2)}}{\tau^{(0)}},
\end{equation}
which shows that $\frac{\tau^{(n+1)}}{\tau^{(0)}}$ and
$\frac{\tau^{(n+2)}}{\tau^{(0)}}$ have the different signs. As a
result, $\frac{\tau^{(n)}}{\tau^{(0)}}$ and
$\frac{\tau^{(n+2)}}{\tau^{(0)}}$ have the different signs. So we
can prove that
\begin{equation}\label{anexpression}
    a(n)=\frac{n(n-1)}{2}
\end{equation}
by induction.

When $n$ is odd, the expression of the transformed BKP's own tau
function $\tau$ is some complicated, which will be the
main task of this section. For this, the first thing is to
root the quadratic form of $X^T\Omega^*X$, thus we need
to show that the quadratic form of $X^T\Omega^*X$ is positive
semidefinte. In fact, from the fact ${\rm det}(\Omega)=0$, we can
know that ${\rm rank}(\Omega)\leq n-1$, and thus ${\rm
rank}(\Omega^*)\leq 1$ \cite{wangqingwen}. Further according to the fact
$\tau_{KP}^{(n)}\neq 0$ and (\ref{bkptaubkpgaugen3}), ${\rm
rank}(\Omega)= n-1$ and ${\rm rank}(\Omega^*)= 1$. Thus for the
$n\times n$ antisymmetric matrix $\Omega$, there exists the
orthogonal matrix $Q$\cite{wangqingwen}, such that,
\begin{equation}\label{omegaiden1}
    \Omega=Q^TBQ
\end{equation}
where
\begin{equation}\label{bexpression}
    B=\left(
                \begin{array}{cccccc}
                  0 & b_1 &  &  &  & \\
                  -b_1 & 0 &  &  &  & \\
                   & & \ddots &  &  &\\
                   &  &  & 0 & b_s & \\
                   &  &  & -b_{s} & 0 &\\
                   &  &  &  &  & 0\\
                \end{array}
              \right),
\end{equation}
with $s=(n-1)/2$. Therefore,
\begin{equation}\label{omegastariden1}
    \Omega^*=Q^{*}B^*Q^{*T}=B_{nn}\alpha_n\alpha_n^T=\beta\beta^T,
\end{equation}
with
\begin{equation}\label{bstar}
    B^*=\left(
          \begin{array}{ccc}
            0 &  &  \\
             & \ddots &  \\
             &  & B_{nn} \\
          \end{array}
        \right),
\end{equation}
where $B_{nn}=b_1^2\cdots b_s^2\geq0$,
$Q^*=(\alpha_1,...,\alpha_n)$, and $\beta=b_1\cdots b_s\alpha_n$. So
the matrix $\Omega^*$ is positive semidefinite and
\begin{equation}\label{oddomega}
    \sqrt{X^T\Omega^*X}=X^T\beta.
\end{equation}

On the other hand, denote $D\left[\begin{array}{c}i \\j
\end{array}\right]$ as the $( j, k)$th minor of $N$th-order
determinant $D={\rm det}(a_{ij})_{1\leq i,j\leq N}$, that is, the
$(N-1)$th-order determinant obtained by eliminating the $j$th row
and the $k$th column from $D$, and the $(N-2)$nd-order determinant
obtained by eliminating the $j$th and $k$th rows and the $l$th and
$m$th columns from the determinant $D$ as $D\left[\begin{array}{cc}
j & k \\l & m\end{array} \right]$, then the Jacobi identity \cite{hirotabook}
is expressed in the following form
\begin{equation}\label{jacobi}
D\left[\begin{array}{c}i \\i
\end{array}\right]D\left[\begin{array}{c}j \\j
\end{array}\right]-D\left[\begin{array}{c}i \\j
\end{array}\right]D\left[\begin{array}{c}j \\i
\end{array}\right]=D\left[\begin{array}{cc}
i & j \\i & j\end{array} \right]D.
\end{equation}

Thus for the antisymmetric matrix $\Omega$ of odd order, from
(\ref{jacobi}) and the fact that $\Omega^*$ is symmetric, we have
\begin{equation}\label{omegaiden2}
D\left[\begin{array}{c}i \\j
\end{array}\right]^2=D\left[\begin{array}{c}i \\i
\end{array}\right]D\left[\begin{array}{c}j \\j
\end{array}\right].
\end{equation}
Since $D\left[\begin{array}{c}i \\i
\end{array}\right]$ is the determinant of the antisymmetric matrix of even order, we have
\begin{equation}\label{dii}
    D\left[\begin{array}{c}i \\i
\end{array}\right]={\rm
Pf}^2(\Phi_1,...,\widehat{\Phi_i},...,\Phi_n)\triangleq{\rm Pf}_i^2,
\ \ i=1,2,...,n.
\end{equation}
Therefore if assume $\Omega^*=(A_{ij})_{1\leq i,j\leq n}$, then
\begin{equation}\label{omegastariden2}
    A_{ij}=(-1)^{i+j}D\left[\begin{array}{c}j \\i
\end{array}\right]=(-1)^{i+j}{\rm Pf}_i {\rm Pf}_j,
\end{equation}
where we have used the fact $\Omega^*$ is positive semidefinte to
omit the negative sign when taking the square root. So we can take
\begin{equation}\label{betaexpression}
    \beta=(-1)^{b(n)+1}(-{\rm
Pf}_1,{\rm Pf}_2,-{\rm Pf}_3...,{\rm Pf}_{n-1},-{\rm Pf}_n)^T.
\end{equation}
Thus by (\ref{oddomega}),
\begin{eqnarray}
\sqrt{X^T\Omega^*X}&=&(-1)^{b(n)+1}\sum_{j=1}^n(-1)^{j}\Phi_j{\rm
Pf}_j=(-1)^{b(n)}\sum_{j=1}^n(-1)^{j-1}\Omega(\Phi_j,1){\rm
Pf}_j\nonumber\\
&=&(-1)^{b(n)}\sum_{j=1}^n(-1)^{j-1}{\rm Pf}(\Phi_j,1){\rm
Pf}_j=(-1)^{b(n)}{\rm
Pf}(\Phi_1,\Phi_2,...,\Phi_n,1).\label{sqrtomega}
\end{eqnarray}
Therefore according to (\ref{taurelation}), (\ref{bkptaubkpgaugen3})
and (\ref{sqrtomega}),
\begin{equation}\label{bkptaubkpgaugenodd}
\frac{\tau^{(n)}}{\tau^{(0)}}= \frac{(-1)^{b(n)}}{2^{n/2}}{\rm
Pf}(\Phi_1,\Phi_2,...,\Phi_n,1),
\end{equation}
when $n$ is odd. By the same way as the even case, we have
\begin{equation}\label{bnexpression}
    b(n)=\frac{n(n-1)}{2}.
\end{equation}

At last, we summarize the result above as the following proposition.
\begin{proposition}\label{bkpgaugen}
Under the gauge transformation $T_n$ (see (\ref{tngaugebkp})), the
BKP tau function $\tau$ and the eigenfunction $\chi$ will transform
as
\begin{itemize}
  \item $n$ is even
\begin{eqnarray}
\frac{\tau^{(n)}}{\tau^{(0)}}&=&\frac{(-1)^{\frac{n(n-1)}{2}}}{2^{n/2}}{\rm
Pf}(\Phi_1,\Phi_2,...,\Phi_n),\label{bkptaubkpgaugenneweven}\\
\chi^{(n)}&=&\frac{{\rm Pf}(\Phi_1,\Phi_2,...,\Phi_n,\chi,1)}{{\rm
Pf}(\Phi_1,\Phi_2,...,\Phi_n)};\label{bkpeigenbkpgaugenneweven}
\end{eqnarray}
  \item $n$ is odd
\begin{eqnarray}
\frac{\tau^{(n)}}{\tau^{(0)}}&=&\frac{(-1)^{\frac{n(n-1)}{2}}}{2^{n/2}}{\rm
Pf}(\Phi_1,\Phi_2,...,\Phi_n,1),\label{bkptaubkpgaugennewodd}\\
\chi^{(n)}&=&-\frac{{\rm Pf}(\Phi_1,\Phi_2,...,\Phi_n,\chi)}{{\rm
Pf}(\Phi_1,\Phi_2,...,\Phi_n,1)}.\label{bkpeigenbkpgaugennewodd}
\end{eqnarray}
\end{itemize}
\end{proposition}

\section{the derivation of the Fay like identities for the BKP Hierarchy by the gauge transformation}
In this section, we will investigate the applications of the gauge
transformation in the derivation of the Fay-like identities for the
BKP hierarchy.

Firstly, with the help of the spectral representation
(\ref{spectralrepresentation}) for the eigenfunction, we have
\begin{equation}\label{phi0spectralrepresentation}
\Phi_{k+1}^{(0)}=\int
d\lambda\varphi_{k+1}^{(0)}(\lambda)\psi_{BA}^{(0)}(t,\lambda),
\end{equation}
where $\Phi_{k+1}^{(0)}$ is the generators of the gauge
transformation $T_n$ (see (\ref{tngaugebkp})). Then the application
of the gauge transformation $T_k$ to the both side of
(\ref{phi0spectralrepresentation}) will lead to
\begin{equation}\label{phikspectralrepresentation}
\Phi_{k+1}^{(k)}=\int
d\lambda\varphi_{k+1}^{(0)}(\lambda)\psi_{BA}^{(k)}(t,\lambda)=\int
d\lambda\varphi_{k+1}^{(0)}(\lambda)\frac{\tau^{(k)}(t-2[\lambda^{-1}])}{\tau^{(k)}(t)}e^{\xi(t,\lambda)},
\end{equation}
by using (\ref{tau}) and (\ref{bkpwavegauge1}). Further according to
(\ref{bkptaubkpgaugen}), we can get the following recurrence
relation,
\begin{eqnarray}
\tau^{(k+1)}&=&\frac{1}{\sqrt{2}}\int
d\lambda\varphi_{k+1}^{(0)}(\lambda)e^{\xi(t,\lambda)}\tau^{(k)}(t-2[\lambda^{-1}])\nonumber\\
&=&\frac{1}{\sqrt{2}}\int
d\lambda\varphi_{k+1}^{(0)}(\lambda):e^{-\hat{\theta}(\lambda)}:\tau^{(k)}(t)\label{taurecurrence}
\end{eqnarray}
where
\begin{equation}\label{vertexoperator}
\hat{\theta}(\lambda)=-\sum_{i=0}^\infty
t_{2i+1}\lambda^{2i+1}+\sum_{i=0}^\infty\frac{2}{(2i+1)\lambda^{2i+1}}\frac{\pa}{\pa
t_{2i+1}},
\end{equation}
and $:...:$ indicates standard normal ordering, i.e., $\frac{\pa
}{\pa t_{2i+1}}$ is on the right of $t_{2i+1}$. Using the following
identity:
\begin{equation}\label{wickidentity}
    :e^{-\hat{\theta}(\lambda_k)}:...:e^{-\hat{\theta}(\lambda_0)}:=:
e^{-\sum_{j=0}^k\hat{\theta}(\lambda_k)}:\prod_{i>j}\left(\frac{\lambda_i-\lambda_j}{\lambda_i+\lambda_j}\right),
\end{equation}
we can solve the recurrence relation (\ref{taurecurrence}) and
express $\tau^{(k+1)}$ in terms of the initial tau function
$\tau^{(0)}$:
\begin{eqnarray}
\tau^{(k+1)}&=&\left(\frac{1}{2}\right)^{\frac{k+1}{2}}\int
\prod_{j=0}^kd\lambda_j:e^{-\hat{\theta}(\lambda_k)}:...:e^{-\hat{\theta}(\lambda_0)}:\tau^{(0)}(t)\prod_{j=0}^k\varphi_{j+1}^{(0)}(\lambda_j)\nonumber\\
&=&\left(\frac{1}{2}\right)^{\frac{k+1}{2}}\int
\prod_{j=0}^kd\lambda_j\prod_{i>j}\left(\frac{\lambda_i-\lambda_j}{\lambda_i+\lambda_j}\right)
\prod_{s=0}^kf_s(t,\lambda_s)\tau^{(0)}(t-2\sum_{j=0}^k[\lambda_j^{-1}]),\label{taukplus1result1}
\end{eqnarray}
where
\begin{equation}\label{fkexpress}
    f_k(t,\lambda)=\varphi_{k+1}^{(0)}(\lambda)e^{\xi(t,\lambda)}.
\end{equation}

 On the other hand, when $k$ is odd, if let $k+1=2n$, then according to Proposition \ref{bkpgaugen} and (\ref{bsepgeneralformula}),
\begin{eqnarray}
&&\frac{\tau^{(k+1)}}{\tau^{(0)}}\nonumber\\
&=&\frac{(-1)^{\frac{k(k+1)}{2}}}{2^{{(k+1)}/2}}{\rm
Pf}(\Phi_1,\Phi_2,...,\Phi_{k+1})=\frac{(-1)^{\frac{k(k+1)}{2}}}{2^{{(k+1)}/2}}\sum_P(-1)^P\Omega_{i_1,i_2}\Omega_{i_3,i_4}...\Omega_{i_{2n-1},i_{2n}}\nonumber\\
&=&\frac{(-1)^{\frac{k(k+1)}{2}}}{2^{{(k+1)}/2}}\int
\prod_{j=0}^kd\lambda_j\prod_{s=0}^kf_s(t,\lambda_s)\sum_P(-1)^P\prod_{l=1}^n
\left(\frac{\lambda_{i_{2l-1}}-\lambda_{i_{2l}}}{\lambda_{i_{2l-1}}+\lambda_{i_{2l}}}\right)
\frac{\tau^{(0)}(t-2[\lambda_{i_{2l-1}}^{-1}]-2[\lambda_{i_{2l}}^{-1}])}{\tau^{(0)}(t)}\nonumber\\
&=&\frac{(-1)^{\frac{k(k+1)}{2}}}{2^{{(k+1)}/2}}\int
\prod_{j=0}^kd\lambda_j\prod_{s=0}^kf_s(t,\lambda_s){\rm
Pf}\left(\left(\frac{\lambda_{i}-\lambda_{j}}{\lambda_{i}+\lambda_{j}}\right)
\frac{\tau^{(0)}(t-2[\lambda_{i}^{-1}]-2[\lambda_{j}^{-1}])}{\tau^{(0)}(t)}\right)_{0\leq
i,j\leq k},\label{oddtaukplus1result2}
\end{eqnarray}
where $P$ is the sequence selected from \{$0,1,2,3,...,2n-1$\} that
satisfy
\begin{eqnarray}
&&i_1<i_2,i_3<i_4,i_5<i_6,...,i_{2n-1}<i_{2n},\nonumber\\
&&i_1<i_3<i_5<...<i_{2n-1},\label{sequence}
\end{eqnarray}
and the factor $(-1)^P$ takes the value $+1$ ($-1$) if the sequence
$i_1, i_2, ..., i_{2n}$ is an even (odd) permutation of $0,1, 2,...,
2n-1.$ Thus comparing with (\ref{taukplus1result1}), we can get
\begin{eqnarray}
&&\prod_{i>j}\left(\frac{\lambda_i-\lambda_j}{\lambda_i+\lambda_j}\right)
\tau^{(0)}(t-2\sum_{j=0}^k[\lambda_j^{-1}])\nonumber\\
&=&(-1)^{\frac{k(k+1)}{2}}\tau^{(0)}(t){\rm
Pf}\left(\left(\frac{\lambda_{i}-\lambda_{j}}{\lambda_{i}+\lambda_{j}}\right)
\frac{\tau^{(0)}(t-2[\lambda_{i}^{-1}]-2[\lambda_{j}^{-1}])}{\tau^{(0)}(t)}\right)_{0\leq
i,j\leq k}.\label{fay1}
\end{eqnarray}

When $k$ is even, Proposition \ref{bkpgaugen},
(\ref{spectralrepresentation}) and (\ref{bsepgeneralformula}) will
lead to
\begin{eqnarray}
&&\frac{\tau^{(k+1)}}{\tau^{(0)}}\nonumber\\
&=&\frac{(-1)^{\frac{k(k+1)}{2}}}{2^{{(k+1)}/2}}{\rm
Pf}(\Phi_1,\Phi_2,...,\Phi_{k+1},1)=\frac{(-1)^{\frac{k(k+1)}{2}}}{2^{{(k+1)}/2}}\sum_{j=1}^{k+1}(-1)^{j-1}\Phi_j{\rm Pf}_j\nonumber\\
&=&\frac{(-1)^{\frac{k(k+1)}{2}}}{2^{{(k+1)}/2}}\int
\prod_{j=0}^kd\lambda_j\prod_{s=0}^kf_s(t,\lambda_s)\sum_{l=0}^{k}(-1)^{l}\frac{\tau^{(0)}(t-2[\lambda_{l}^{-1}])}{\tau^{(0)}(t)}\nonumber\\
&&\times{\rm
Pf}\left(\left(\frac{\lambda_{p}-\lambda_{q}}{\lambda_{p}+\lambda_{q}}\right)
\frac{\tau^{(0)}(t-2[\lambda_{p}^{-1}]-2[\lambda_{q}^{-1}])}{\tau^{(0)}(t)}\right)_{0\leq
p,q\leq k,p\neq l,q\neq l}\label{eventaukplus1result2}
\end{eqnarray}
The comparison with (\ref{taukplus1result1}) gives rise to
\begin{eqnarray}
&&\prod_{i>j}\left(\frac{\lambda_i-\lambda_j}{\lambda_i+\lambda_j}\right)
\tau^{(0)}(t-2\sum_{j=0}^k[\lambda_j^{-1}])=(-1)^{\frac{k(k+1)}{2}}\sum_{l=0}^{k}(-1)^{l}\tau^{(0)}(t-2[\lambda_{l}^{-1}])\nonumber\\
&&\times{\rm
Pf}\left(\left(\frac{\lambda_{p}-\lambda_{q}}{\lambda_{p}+\lambda_{q}}\right)
\frac{\tau^{(0)}(t-2[\lambda_{p}^{-1}]-2[\lambda_{q}^{-1}])}{\tau^{(0)}(t)}\right)_{0\leq
p,q\leq k,p\neq l,q\neq l}.\label{fay2}
\end{eqnarray}
We summarize the result above into the following proposition
\begin{proposition}\label{faygauge}
The initial tau function $\tau^{(0)}(t)$ for the BKP hierarchy
satisfies the higher Fay-like identities.
\begin{itemize}
  \item $k$ is odd
\begin{eqnarray}
&&\prod_{i>j}\left(\frac{\lambda_i-\lambda_j}{\lambda_i+\lambda_j}\right)
\tau^{(0)}(t-2\sum_{j=0}^k[\lambda_j^{-1}])\nonumber\\
&=&(-1)^{\frac{k(k+1)}{2}}\tau^{(0)}(t){\rm
Pf}\left(\left(\frac{\lambda_{i}-\lambda_{j}}{\lambda_{i}+\lambda_{j}}\right)
\frac{\tau^{(0)}(t-2[\lambda_{i}^{-1}]-2[\lambda_{j}^{-1}])}{\tau^{(0)}(t)}\right)_{0\leq
i,j\leq k}.\label{propfay1}
\end{eqnarray}
  \item $k$ is even
\begin{eqnarray}
&&\prod_{i>j}\left(\frac{\lambda_i-\lambda_j}{\lambda_i+\lambda_j}\right)
\tau^{(0)}(t-2\sum_{j=0}^k[\lambda_j^{-1}])=(-1)^{\frac{k(k+1)}{2}}\sum_{l=0}^{k}(-1)^{l}\tau^{(0)}(t-2[\lambda_{l}^{-1}])\nonumber\\
&&\times{\rm
Pf}\left(\left(\frac{\lambda_{p}-\lambda_{q}}{\lambda_{p}+\lambda_{q}}\right)
\frac{\tau^{(0)}(t-2[\lambda_{p}^{-1}]-2[\lambda_{q}^{-1}])}{\tau^{(0)}(t)}\right)_{0\leq
p,q\leq k,p\neq l,q\neq l}.\label{propfay2}
\end{eqnarray}
\end{itemize}
\end{proposition}
At last, let's conclude this section with some
examples:
\begin{itemize}
  \item For $k=2$, (\ref{propfay2}) will lead to
\begin{eqnarray}
&&\frac{(\lambda_2-\lambda_0)(\lambda_2-\lambda_1)(\lambda_1-\lambda_0)}{(\lambda_2+\lambda_0)(\lambda_2+\lambda_1)(\lambda_1+\lambda_0)}
\tau^{(0)}(t-2[\lambda_0^{-1}]-2[\lambda_1^{-1}]-2[\lambda_2^{-1}])\tau^{(0)}(t)\nonumber\\
&&=\frac{(\lambda_{2}-\lambda_{1})}{(\lambda_{2}+\lambda_{1})}\tau^{(0)}(t-2[\lambda_{1}^{-1}]-2[\lambda_{2}^{-1}])\tau^{(0)}(t-2[\lambda_{0}^{-1}])\nonumber\\
&&-\frac{(\lambda_{2}-\lambda_{0})}{(\lambda_{2}+\lambda_{0})}\tau^{(0)}(t-2[\lambda_{0}^{-1}]-2[\lambda_{2}^{-1}])\tau^{(0)}(t-2[\lambda_{1}^{-1}])\nonumber\\
&&+\frac{(\lambda_{1}-\lambda_{0})}{(\lambda_{1}+\lambda_{0})}\tau^{(0)}(t-2[\lambda_{1}^{-1}]-2[\lambda_{0}^{-1}])\tau^{(0)}(t-2[\lambda_{2}^{-1}]),\label{fayk2}
\end{eqnarray}
which is just (\ref{Fayeq}) for
$s_0=-\lambda_0^{-1},s_1=-\lambda_1^{-1},s_2=-\lambda_2^{-1},s_3=0$.
  \item For $k=3$, according to (\ref{propfay1}),
\begin{eqnarray}
&&\frac{(\lambda_3-\lambda_0)(\lambda_3-\lambda_1)(\lambda_3-\lambda_2)(\lambda_2-\lambda_0)(\lambda_2-\lambda_1)(\lambda_1-\lambda_0)}
{(\lambda_3+\lambda_0)(\lambda_3+\lambda_1)(\lambda_3+\lambda_2)(\lambda_2+\lambda_0)(\lambda_2+\lambda_1)(\lambda_1+\lambda_0)}\\
&&\times
\tau^{(0)}(t-2[\lambda_0^{-1}]-2[\lambda_1^{-1}]-2[\lambda_2^{-1}]-2[\lambda_3^{-1}])\tau^{(0)}(t)\nonumber\\
&&=\frac{(\lambda_{0}-\lambda_{1})(\lambda_{2}-\lambda_{3})}{(\lambda_{0}+\lambda_{1})(\lambda_{2}+\lambda_{3})}\tau^{(0)}(t-2[\lambda_{0}^{-1}]-2[\lambda_{1}^{-1}])\tau^{(0)}(t-2[\lambda_{2}^{-1}]-2[\lambda_{3}^{-1}])\nonumber\\
&&-\frac{(\lambda_{0}-\lambda_{2})(\lambda_{1}-\lambda_{3})}{(\lambda_{0}+\lambda_{2})(\lambda_{1}+\lambda_{3})}\tau^{(0)}(t-2[\lambda_{0}^{-1}]-2[\lambda_{2}^{-1}])\tau^{(0)}(t-2[\lambda_{1}^{-1}]-2[\lambda_{3}^{-1}])\nonumber\\
&&+\frac{(\lambda_{0}-\lambda_{3})(\lambda_{1}-\lambda_{2})}{(\lambda_{0}+\lambda_{3})(\lambda_{1}+\lambda_{2})}\tau^{(0)}(t-2[\lambda_{0}^{-1}]-2[\lambda_{3}^{-1}])\tau^{(0)}(t-2[\lambda_{1}^{-1}]-2[\lambda_{2}^{-1}]),\label{fayk1}
\end{eqnarray}
which is just (\ref{Fayeq}) for
$s_0=-\lambda_0^{-1},s_1=-\lambda_1^{-1},s_2=-\lambda_2^{-1},s_3=-\lambda_3^{-1}$.
\end{itemize}
\section{the additional symmetry and the gauge transformation for the BKP Hierarchy }
In this section, we will study the compatibility between the
additional symmetry and the gauge transformation of the BKP
hierarchy.

Firstly, the additional symmetry of the BKP hierarchy (\ref{laxeq})
is defined by introducing the additional variables $\hat t_{ml}$\cite{Tu07}:
\begin{equation}\label{addsymm}
    \hat \pa_{ml}L=[-(A_{ml}(M,L))_-,L],\quad\hat \pa_{ml}W=-(A_{ml}(M,L))_-W,
\end{equation}
where
\begin{equation}\label{Aml}
    A_{ml}(M,L)=M^mL^l-(-1)^lL^{l-1}M^mL,
\end{equation}
satisfying
\begin{equation}\label{amlrelation}
    A_{ml}(M,L)^*=-\pa A_{ml}(M,L)\pa^{-1}
\end{equation}
and $M$ is the Orlov-Schulman operator defined by
\begin{equation}\label{osoperator}
    M=W\Gamma W^{-1},\quad \Gamma=\sum_{n=0}(2n+1)t_{2n+1}\pa^{2n},
\end{equation}
which satisfies
\begin{equation}\label{mrelation}
    \pa_{2n+1}M=[B_{2n+1},M],\quad [L,M]=1.
\end{equation}
Before considering the additional symmetry under the gauge
transformation, some lemmas are needed.
\begin{lemma}
For any pseudo-differential operator $P$ and a function $f$,
\begin{eqnarray}
(f\pa f^{-1}Pf\pa^{-1}f^{-1})_+&=&f\pa f^{-1}(P)_+f\pa^{-1}f^{-1}-( f\pa f^{-1}P_+)(f)\pa^{-1}f^{-1}\label{type1formula}\\
(f^{-1}\pa^{-1}fPf^{-1}\pa f)_-&=&f^{-1}\pa^{-1}f(P)_-f^{-1}\pa
f-f^{-1}\pa^{-1}(f\pa f^{-1}P^*_+)(f)\label{type2formula}
\end{eqnarray}
\end{lemma}

\begin{lemma}
For the eigenfunction $\Phi$ of the BKP hierarchy,
\begin{equation}\label{amlrelationlemma}
    A_{ml}(M,L)_+(\Phi)=-(T_D(\Phi)A_{ml}(M,L)T_D(\Phi)^{-1})^*_+(\Phi)
\end{equation}
\end{lemma}
\begin{proof}
According to (\ref{amlrelation}) and (\ref{type1formula}),
\begin{eqnarray*}
&&(T_D(\Phi)A_{ml}T_D(\Phi)^{-1})_+^*\\
&=&\Big(T_D(\Phi)(A_{ml})_+T_D(\Phi)^{-1}-\left(T_D(\Phi)(A_{ml})_+\right)(\Phi)\pa^{-1}\Phi^{-1}\Big)^*\\
&=&T_I(\Phi)(A_{ml})^*_+T_I(\Phi)^{-1}+\Phi^{-1}\pa^{-1}\left(T_D(\Phi)(A_{ml})_+\right)(\Phi)\\
&=&-T_I(\Phi)\pa(A_{ml})_+\pa^{-1}T_I(\Phi)^{-1}+\Phi^{-1}\pa^{-1}\left(T_D(\Phi)(A_{ml})_+\right)(\Phi),
\end{eqnarray*}
then
\begin{eqnarray*}
\Big(-T_I(\Phi)\pa(A_{ml})_+\pa^{-1}T_I(\Phi)^{-1}\Big)(\Phi)=-2\Phi^{-1}\int
\Phi\pa_x\Big((A_{ml})_+(\Phi)\Big)
\end{eqnarray*}
and
\begin{eqnarray*}
&&\Big(\Phi^{-1}\pa^{-1}\left(T_D(\Phi)(A_{ml})_+\right)(\Phi)\Big)(\Phi)=\Phi^{-1}\int
\Phi^2\pa_x\Big(\Phi^{-1}(A_{ml})_+(\Phi)\Big)\\
&=&-\Phi^{-1}\int(A_{ml})_+(\Phi)\Phi_x+\Phi^{-1}\int
\Phi\pa_x\Big((A_{ml})_+(\Phi)\Big).
\end{eqnarray*}
Therefore
\begin{eqnarray*}
&&(T_D(\Phi)A_{ml}(M,L)T_D(\Phi)^{-1})^*_+(\Phi)\\
&=&-\Phi^{-1}\int\left((A_{ml})_+(\Phi)\Phi_x+\Phi\pa_x\Big((A_{ml})_+(\Phi)\Big)\right)=-(A_{ml})_+(\Phi)
\end{eqnarray*}

\end{proof}
After the preparation above, we can get the proposition below.
\begin{proposition}\label{addgauge}
Additional symmetry flows (\ref{addsymm}) commute with the gauge transformation $T=T(\Phi)$ for the BKP hierarchy, that is,
\begin{equation}\label{addsymmgaugetrans}
\hat \pa_{ml}\widetilde{L}=[-(A_{ml}(\widetilde{M},\widetilde{L}))_-,\widetilde{L}],\quad \pa_{ml}\widetilde{W}=-(A_{ml}(\widetilde{M},\widetilde{L}))_-\widetilde{W} ,
\end{equation}
where
\begin{equation}\label{lmtransform}
  \widetilde{L}=TLT^{-1},\quad \widetilde{M}=TMT^{-1},
\end{equation}
if and only if the eigenfunction $\Phi$ transforms under the additional symmetries as:
\begin{equation}\label{addactoneigen}
  \hat \pa_{ml}\Phi=(A_{ml}(M,L))_+\Phi.
\end{equation}

\end{proposition}

\begin{proof}
Firstly, by (\ref{addsymm})
\begin{eqnarray}
\hat \pa_{ml}\widetilde{L}&=&\hat \pa_{ml}T\cdot L T^{-1}+T\hat \pa_{ml} L \cdot T^{-1}-TL T^{-1}\cdot\hat \pa_{ml}T\cdot T^{-1}\\
&=&[-T(A_{ml}(M,L))_-T^{-1}+\pa_{ml}T\cdot T^{-1},\widetilde{L}]
\end{eqnarray}
On the other hand,
\begin{eqnarray}
[-(A_{ml}(\widetilde{M},\widetilde{L}))_-,\widetilde{L}]=[-(T A_{ml}(\widetilde{M},\widetilde{L})T^{-1})_-,\widetilde{L}].
\end{eqnarray}
Thus (\ref{addsymmgaugetrans}) holds if and only if
\begin{equation}\label{conditions}
-(T A_{ml}(M,L)T^{-1})_-=-T(A_{ml}(M,L))_-T^{-1}+\hat \pa_{ml}T\cdot T^{-1}
\end{equation}
Since $T=T_I(\Phi/2)T_D(\Phi)=T_I(\Phi)T_D(\Phi)$, thus from (\ref{type1formula}), (\ref{type2formula}), and (\ref{amlrelationlemma})
\begin{eqnarray}
&&-(T A_{ml}(M,L)T^{-1})_-\nonumber\\
&=&-T_I(\Phi)(T_D(\Phi)A_{ml}(M,L)T_D(\Phi)^{-1})_-T_I(\Phi)^{-1}\nonumber\\
&&+\Phi^{-1}\pa^{-1}\Big(T_D(\Phi)(T_D(\Phi)A_{ml}(M,L)T_D(\Phi)^{-1})^*_+\Big)(\Phi)\nonumber\\
&=&-T_I(\Phi)(T_D(\Phi)A_{ml}(M,L)T_D(\Phi)^{-1})_-T_I(\Phi)^{-1}-\Phi^{-1}\pa^{-1}\Big(T_D(\Phi)(A_{ml}(M,L))_+\Big)(\Phi)\nonumber\\
&=&-T(A_{ml}(M,L))_-T^{-1}-T_I(\Phi)\Big(T_D(\Phi)(A_{ml}(M,L))_+\Big)(\Phi)\pa^{-1}\Phi^{-1}T_I(\Phi)^{-1}\nonumber\\
&&-\Phi^{-1}\pa^{-1}\Big(T_D(\Phi)(A_{ml}(M,L))_+\Big)(\Phi).\label{addgaugeproof1}
\end{eqnarray}
On the other hand,
\begin{eqnarray}
\hat\pa_{ml}T\cdot T^{-1}&=&\hat\pa_{ml}T_I(\Phi)\cdot T_I(\Phi)^{-1}+T_I(\Phi)\hat\pa_{ml}T_D(\Phi)\cdot T_D(\Phi)^{-1}T_I(\Phi)^{-1}\nonumber\\
&=&-\Phi^{-1}\hat\pa_{ml}\Phi+\Phi^{-1}\pa^{-1}\Big(\hat\pa_{ml}\Phi\Big) \Phi^{-1}\pa \Phi+T_I(\Phi)\cdot\Big(\hat\pa_{ml}\Phi\Big)\Phi^{-1}\cdot T_I(\Phi)^{-1}\nonumber\\
&&-T_I(\Phi)\cdot\Phi\pa\Phi^{-2}\cdot\Big(\hat\pa_{ml}\Phi\Big)\Phi\pa^{-1}\Phi^{-1}\cdot T_I(\Phi)^{-1}\nonumber\\
&=&-\Phi^{-1}\pa^{-1}\left(T_D(\Phi)\big(\hat\pa_{ml}\Phi\big) \right)-T_I(\Phi)\cdot\left(T_D(\Phi)\big(\hat\pa_{ml}\Phi\big) \right)\pa^{-1}\Phi^{-1}\cdot T_I(\Phi)^{-1}\label{addgaugeproof2}
\end{eqnarray}
Therefore from (\ref{addgaugeproof1}) and (\ref{addgaugeproof2}),  (\ref{conditions}) is true if and only if
\begin{equation*}
  \hat \pa_{ml}\Phi=(A_{ml}(M,L))_+\Phi.
\end{equation*}
\end{proof}
\section{Conclusions and Discussions}
We have provided four  applications of the gauge transformation for
the BKP hierarchy.
\begin{itemize}
  \item The orbit of the gauge transformation for the constrained BKP hierarchy
defines a special $(2 +1)$-dimensional Toda lattice equation structure (see Proposition \ref{orbitgauge});
  \item Starting from the Grammian determinant solutions of the BKP hierarchy, we derived
two types of the Pfaffian structures for the BKP hierarchy (see Proposition \ref{bkpgaugen});
  \item Through the gauge transformation of the BKP hierarchy, the higher order Fay-like identities are derived (see Proposition \ref{faygauge});
  \item The compatibility between the additional symmetry and the gauge transformation of the BKP
hierarchy is studied (see Proposition \ref{addgauge}).
\end{itemize}
Our results show that the gauge transformation is also  convenient
tool to study integrable properties, such as the Fay-like identities and the additional symmetries for the BKP
hierarchy.

For the non-Hermitian random matrix models, there are the Pfaffian
tau function structures\cite{asm99,avm2001}, which are closely
related with the BKP hierarchy. We hope that our results will be
helpful for the comprehension of the link between the BKP hierarchy
and the non-Hermitian random matrix models.

{\bf {Acknowledgements:}}
  {\small   This work is supported by the NSFC (Grant No. 11226196) and ``the Fundamental Research Funds for the
Central Universities" No. 2012QNA45.}

\end{document}